\documentclass[12pt,fleqn]{article}
\usepackage{amsfonts} 
\usepackage{amssymb}
\usepackage{amsmath}
\usepackage{amstext,ifthen,multirow,amsthm} 
\usepackage{hyperref, natbib}
\usepackage{enumitem}
\usepackage{tikz}
\usepackage[final]{pdfpages}
\usepackage{ragged2e}

\renewcommand{\mathbf}{\boldsymbol}

\newtheorem{theorem}{Theorem}
\newtheorem{lemma}{Lemma}

\newtheorem{definition}{Definition}

\def\monthname{\ifcase\month\or
  January\or February\or March\or April\or May\or June\or July\or
  August\or September\or October\or November\or December\fi}
\numberwithin{equation}{section}

\newcommand{\bge}{\begin{equation}}
\newcommand{\ene}{\end{equation}}

\newcommand{\BIT}{\begin{itemize}}
\newcommand{\EIT}{\end{itemize}}
\newcommand{\BNUM}{\begin{enumerate}}
\newcommand{\ENUM}{\end{enumerate}}

\def\mrm#1{\mathrm{#1}}
\def\reals{\mathbb{R}} 
\renewcommand{\exp}[1]{\operatorname{exp}\left(#1\right)} 

\def\E{E} 
\def\V{V} 
\def\P{P} 

\def\Parg#1{\P\left({#1}\right)}

\def\cov{\mrm{cov}} 

\def\Convprob{\overset{p}{\longrightarrow}} 
\def\plim{\mathrm{plim}} 

\newtheorem{corollary}{Corollary}
\newtheorem{remark}{Remark}

\begin{document}


\title{Estimation of Discrete Choice Models: A Machine Learning Approach\thanks{The paper has been developed as part of graduate studies requirements of the authors. We greatly appreciate the guidance of our primary advisors, Guido Imbens and Ali Yurukoglu, during all stages of the project. We also appreciate the help from Lanier Benkard and Paulo Somaini. Many useful comments from the participants of 2018 California Econometrics Conference, Stanford IO reading group, and Facebook Economics reading group shaped the final version of the paper. All remaining mistakes and typos are our own.}}


\author{Nick Doudchenko\footnote{New York, NY 10011, \texttt{ndoudchenko@gmail.com}}\and Evgeni Drynkin\footnote{Menlo Park, CA 94025, \texttt{e.drynkin@gmail.com}}}

\maketitle

\begin{abstract}
In this paper we propose a new method of estimation for discrete choice demand models when individual level data are available. The method employs a two-step procedure. Step 1 predicts the choice probabilities as functions of the observed individual level characteristics. Step 2 estimates the structural parameters of the model using the estimated choice probabilities at a particular point of interest and the moment restrictions. In essence, the method uses nonparametric approximation (followed by) moment estimation. Hence the name---NAME. We use simulations to compare the performance of NAME with the standard methodology. We find that our method improves precision as well as convergence time. We supplement the analysis by providing the large sample properties of the proposed estimator.
\end{abstract}


\section{Introduction}
The primary goal of applied economic research is informing policy decisions. In many practically important cases experimental analysis is infeasible or prohibitively expensive. In those cases researchers have to rely on other sources of identification. Moreover, some applications require extrapolating outside of the support of the observed data which limits the applicability of highly flexible predictive models which are likely to suffer from overfitting. For instance, the analysis of a potential merger might require predicting the quantities and prices after two firms in a three-firm industry merge. It is quite possible that within the relevant time span the industry has always consisted of three firms. As a result, the training set would not contain the data necessary to obtain a model capable of accurate counterfactual predictions. Traditionally researchers would rely on theory to specify a functional form and extrapolate using the relevant counterfactual values---for instance, three firms instead of two. These issues have limited the adoption of purely predictive methods from statistical and machine learning in economic research.

However, in some cases parts of the overall empirical strategy can be posed as predictions problems. In that case the full force of machine learning methods can be invoked providing flexibility of the functional form and computational efficiency.\footnote{For a review of commonly used prediction methods see, for example, \cite*{friedman2001elements} and \cite*{Murphy2012MachineL}.} In this paper we consider discrete choice demand estimation \citep*{mcfadden1973conditional, berry1995automobile, nevo2000practitioner} with coefficients that depend on observable individual level characteristics. The standard approach utilized in the literature is to assume a specific parametric form of the dependence on the individual level variables. We propose an alternative two-step procedure. First, we solve a prediction problem that links the individual level variables to the choice probabilities. Second, we estimate the standard discrete choice model to find the coefficients at pre-defined values of the covariates. This allows us to obtain the values of the coefficients at any other points by solving a system of linear equations.

The standard approach relies heavily on either correctly specifying the functional form or using a functional form that doesn't lead to a high bias or variance. It doesn't take into account the potential structure of the space, such as, for example, sparsity, either. Another issue with this approach is that it may become computationally burdensome when the dimensionality of the individual level characteristics increases. Using a prediction approach can address all of these issues.

We compare the proposed method to the more common approach by simulating the distribution of estimated elasticities. We consider two variations of the standard procedure: (i) the case when the parametric form of the coefficient as a function of the individual level variable is misspecified, and (ii) the case when it is specified correctly (the oracle case). When compared in terms of the root-mean-square error our method performs about 65\% better than the oracle\footnote{The main reason why our method is more precise than the one based on the correct specification is the complex optimization needed to implement oracle estimation. For some realizations of the data this leads to the divergence of the optimization routine. This can serve as some evidence that despite the fact that asymptotically the behavior of the algorithms should coincide, one should lean towards utilizing NAME.} while the misspecified estimation performs almost 100\% worse than the oracle. Perhaps the biggest advantage of the proposed procedure is its computational efficiency. We find that in the simplest case---when the individual level characteristic is a scalar---our method converges almost twice as fast as the considered alternatives. The efficiency gains become more significant when the dimensionality increases.

We also provide theoretical results that justify the use of the proposed estimator---it is consistent and has the same asymptotic distribution as the oracle estimator.

\section{Related Literature}
The effects of consumer specific tastes have been found to be essential for a variety of industries
including among others: cable television \cite{crawford2012welfare}, new cars \cite{doi:10.1086/379939}, alcohol beverages \cite{griffith2019tax}, and soda \cite{dubois2017well}. \cite*{berry2014identification} and \cite*{dunker2017nonparametric} address the problem of nonparametric identification in discrete choice demand models and \cite*{compiani2018nonparametric} proposes a specific way to estimate a model given the data commonly available in applications. 

A study somewhat related to ours is \cite*{gillen2014demand} that uses ideas similar to those of \cite*{belloni2014inference} and \cite*{farrell2015robust} in the context of demand estimation when the space of product-level characteristics is high-dimensional and sparse. \cite*{gillen2015blp} study the issue of selection from a set of demographic variables to be included in demand estimation. \cite*{athey2018estimating} use individual level data and machine learning methods to estimate demand for restaurants and travel time.

There are a number of studies that apply ideas from machine learning to causal inference. See, for example, \cite*{hartford2016counterfactual, wager2017estimation,athey2016recursive,fan2012adaptive,belloni2011lasso, doudchenko2016balancing, gautier2011high,hansen2014instrumental,chernozhukov2015valid, bloniarz2016lasso,athey2016efficient,chernozhukov2017double,belloni2011inference,belloni2014high,belloni2012sparse}.

\section{Model}
Consider the following discrete choice setting. There are $M$ separate markets populated by $N_m$ individuals for $m=1,\dots,M$. Each individual $i=1,\dots,N_m$ is characterized by a set of observable covariates $Z_{im}\in\reals^p$ and her product choice $d_{im}\in\{0,1,\dots,J\}$, where $J$ is the number of products available in each market\footnote{It is straightforward to generalize to the case when each market has a separate set of $J_m$ products.} and $d_{im}=0$ corresponds to the outside good. Product $j=1,\dots,J$ is characterized by an observable variable $X_{jm}\in\reals^k$ and an unobservable variable $\xi_{jm}(Z_{im})\in\reals$. The dependence of the unobserved quality on individual level characteristics is a plausible modeling assumption. For example, it can reflect market segmentation across the products over different socio-demographics groups. Those differences can be a result of concentration of ad spendings or product positioning to a particular age, income, or any other cohort. We also assume the existence of a set of instrumental variables $W_{jm}\in\reals^l$ that are uncorrelated with $\xi_{jm}(Z_{im})$. Utility $u_{ijm}$ derived by individual $i$ from buying good $j=0,\dots,J$ in market $m$ is given by
\begin{eqnarray*}
  u_{ijm} &=& f(\theta(Z_i);X_{jm},\xi_{jm}(Z_i),\zeta_{ijm}),
\end{eqnarray*}
\noindent where $\zeta_{ijm}$ is an idiosyncratic error. For example, in \cite{berry1994estimating}:
\begin{eqnarray*}
  u_{ijm} &=& \theta X_{jm} + \xi_{jm} + \varepsilon_{ijm},
\end{eqnarray*}
\noindent so $\zeta = \varepsilon$ is an additive term with a Generalized Extreme Value Type-I (Gumbel) distribution, while in \cite*{berry1995automobile}: 
\begin{eqnarray*}
  u_{ijm} &=& (\theta + \nu_{im}) X_{jm} + \xi_{jm} + \varepsilon_{ijm},
\end{eqnarray*}
\noindent so it has an extra component additive to $\theta(Z_i)$, meaning $\zeta = (\varepsilon,\nu)$.

We assume that each person chooses the good that provides the highest level of utility in which case the probability that individual $i$ chooses good $j=1,\dots,J$ in market $m$ is: 
\begin{multline*}
	s_{jm}(Z_{im}) = \\
	\int\mathbb{I}\Big[f(\theta(Z_{im});X_{jm},\xi_{jm}(Z_{im}),\zeta_{ijm})= \\
	\max_{j'}\{f(\theta(Z_{im});X_{j'm},\xi_{j'm}(Z_{im}),\zeta_{ij'm})\}\Big]\,dF_{\zeta}(\zeta_{im}).
\end{multline*}
For example, in \cite{berry1994estimating} this probability is:
\begin{eqnarray*}
  s_{jm}(Z_{im}) &=& \frac{\exp{\beta(Z_{im})^TX_{jm} - \alpha(Z_{im}) P_{jm} + \xi_{jm}(Z_{im})}}{1 + \sum_{j'}\exp{\beta(Z_{im})^TX_{j'm} - \alpha(Z_{im}) P_{j'm} + \xi_{j'm}(Z_{im})}}.
\end{eqnarray*}
\noindent Averaging across individuals we obtain the market shares,
\begin{eqnarray*}
  s_{jm} &=& \frac{1}{N_m}\sum_{i=1}^{N_m}s_{jm}(Z_{im}).
\end{eqnarray*}

\noindent Finally, identification comes form the moment condition:
\begin{eqnarray*}
	\E_{jm}\left[H(\theta(Z_{im}); \xi_{jm}(Z_{im}), X_{jm}, W_{jm})|Z_{im}\right] = 0,\quad \forall Z_{im}.
\end{eqnarray*}

\section{Estimation}
\subsection{Standard Approaches}
There are two standard approaches in the literature. The first approach goes back to \cite{doi:10.1086/379939} and assumes that: (i) $\theta(Z) = g(Z,\gamma)$, (ii) $\xi(Z) = \xi$. With these two assumptions the following algorithm is used:
\begin{enumerate}
  \item Specify the parametric functional form of $\theta(Z)=g(Z,\gamma)$, where $g$ is known and $\gamma$ is a parameter.
  \item Initialize the parameters to be estimated: $\gamma=\gamma_0$.
  \item Iterate $\gamma$ until convergence. At step $n$:
  \begin{enumerate}
    \item Solve for $\xi_{jm}$, $j=1,\dots,J$, $m=1,\dots,M$ by inverting the market shares---find $\hat{\xi}_{jm}$ such that the implied market shares $s_{jm}(\gamma_n,\hat{\xi}_{jm})$ coincide with the observed market shares $s_{jm}$ (BLP inversion).
    \item Compute the loss function $L(\gamma_n,\hat{\xi}_{jm})$ (usually based on generalized method of moments/minimum distance estimation).
    \item Update the parameters $\gamma=\gamma_{n+1}$.
  \end{enumerate}
\end{enumerate}

\noindent An alternative approach, often called bunching,
\footnote{This idea is used in \cite{dubois2017well} and
\cite{dubois2017well} among others.} is to break the data into several pieces based on the values of $Z$. This approach typically does not assume $\xi(Z)=\xi$, however it is restrictive in a sense that $\theta(Z)$ and $\xi(Z)$ are piece-wise constant. The estimation proceeds as follows:
\begin{enumerate}
	\item Specify the bunching regions $\mathcal{Z}_1,\dots,\mathcal{Z}_K$.
	\item For every $k=1,\dots,K$ run the following estimation procedure:
	\begin{enumerate}
		\item Restrict the data to a subsample such that $Z_{im}\in\mathcal{Z}_k$.
		\item Assume that $\theta(Z)=\theta_k$ and $\xi(Z)=\xi_k$ for any $Z\in\mathcal{Z}_k$.
		\item Iterate $\theta_k$ until convergence. At step $n$:
		\begin{enumerate}
			\item Find $\hat{\xi}_{k,jm}$ so that $s_{jm}(\theta_k,\hat{\xi}_{k,jm}) = \hat{s}_{k,jm}$ (BLP inversion).
			\item Update $\theta_k$ to match the moment $\hat{H}(\theta_{k}, \hat{\xi},X,W)$ better.
		\end{enumerate}
	\end{enumerate}
\end{enumerate}


\subsection{Proposed Method}
Both methods described above have limitations. \cite{doi:10.1086/379939}, though continuous in nature, limits the flexibility of unobserved heterogeneity. In addition, as the dimensionality of $\gamma$ grows, the number of moment restrictions provided by $H(\cdot)$ may be insufficient. The researcher thus have to add extra moment conditions which may not be valid. A default option for the moments is the difference between the implied and observed correlation between $Z$ and the characteristics of the $i$'s choice, that is $X_{d_im}$. Adding moments harms estimation in two ways. First, it can lower the precision and introduce the biases if those moments are misspecified. The misspecification can occur due to $g(Z,\gamma)$ not being sufficiently flexible. By increasing flexibility, however, the researcher introduces more and more variables to estimate, slowing the estimation routine.

Bunching, in contrast, is a very scalable approach that requires only a handful of variable transformations and moment restrictions to estimate the parameters of interest. It is also a more flexible approach when it comes to the form of $\xi(Z)$. However, first, bunching is inherently discontinuous. It can be problematic, when the researcher wants to identify subpopulations that respond differently to a particular policy. For example, though it seems natural in most cases that 40 and 41 year olds would respond somewhat similar to a change in the market conditions, this approach might predict very different reactions if those ages fall into different bunching groups. Second, the grouping is a somewhat arbitrary decision often made by the researcher without any fundamental reason. Moreover, if the dimensionality of $Z$ is reasonably high, each group risks to either have only a few observations or be too broad. Though there is a wide variety of machine learning algorithms designed to handle high dimensional prediction problems, it is not immediately clear how to combine them with the bunching approach.

This paper combines modern machine learning approaches with the bunching idea in a specific way that allows one to generate continuous $\theta(Z)$ and $\gamma(Z)$, add any prior assumptions on their shape (for example, sparsity), and use modern techniques to soften the curse of dimensionality issue.

The method we propose in this paper is based on the following idea. If in each unit of observation the market shares for a given $Z=Z_0$ were known, the problem would reduce to a simple discrete choice estimation. The procedure would be similar to the bunching approach described above. In fact, bunching produces predictions for $Z=Z_0$ under the assumption of piece-wise constant functions. As $s_{jm}(Z_0)$ are unobserved, we attempt to estimate them using the individual level data, $(Z_{im},d_{im})$.\footnote{Another issue is estimating $\beta(Z)$ at $Z\neq Z_0$. We address this in the next section.}

We propose the following algorithm, which we call NAME:
\begin{enumerate}
  \item Use a prediction method of choice to fit $s(z,j,m) = \Parg{d_{im}=j\vert Z_{im} = z}$.
  \item Pick $Z_0$ within the support of $Z$.\footnote{The optimal choice of $Z_0$ is an important question that we leave for future research. In the simulations we use the median of the observed values for every dimension of $Z$.}
  \item Obtain the predictions $\hat{s}(Z_0,j,m)$ for every $j=1,\dots,J$ and $m=1,\dots,M$.
  \item Initialize the parameters to be estimated: $\theta=\theta_0$.
  \item Iterate $\theta$ until convergence. At step $n$:
  \begin{enumerate}
    \item Solve for $\xi_{jm}$, $j=1,\dots,J$, $m=1,\dots,M$ by inverting the predicted market shares---find $\hat{\xi}_{jm}$ such that the implied market shares $s_{jm}(\theta_n,\hat{\xi}_{jm})$ coincide with the predicted market shares $\hat{s}_{jm}$.
    \item Compute the loss function $L(\theta_n,\hat{\xi}_{jm})$.
    \item Update the parameters $\theta=\theta_{n+1}$.
  \end{enumerate} 
\end{enumerate}
It can be seen that NAME is a generalization of the bunching. In fact, if $Z$ can only take on finitely many points, the two coincide. A difference between our algorithm and \cite{doi:10.1086/379939} is slightly more subtle. The main advantage in the performance comes from the fact that NAME predicts the individual choices once, while this step is nested inside the optimization routine in \cite{doi:10.1086/379939}. The difference can be seen in Figure \ref{fig:schematic_comparison}.
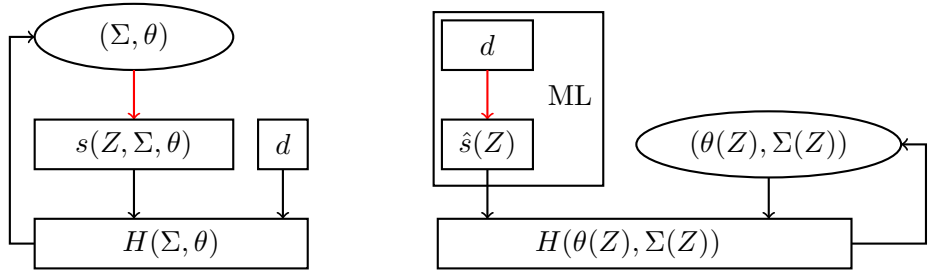
\begin{figure}
	\begin{center}
		\begin{tikzpicture}[scale=1.1]
			\draw[thick] (0,0) ellipse (1.2 and 0.4);
			\node at (0,0) {\small $(\Sigma, \theta)$};
			\draw[thick, ->, red] (0,-0.4) -- (0,-1);
			\draw[thick] (-1.2,-1) rectangle (1.2,-1.6);
			\node at (0,-1.3) {\small $s(Z,\Sigma,\theta)$};
			\draw[thick] (1.5,-1) rectangle (2.1,-1.6);
			\node at (1.8,-1.3) {\small $d$};
			\draw[thick, ->] (0,-1.6) -- (0,-2.2);
			\draw[thick, ->] (1.8,-1.6) -- (1.8,-2.2);
			\draw[thick] (-1.2,-2.2) rectangle (2.1,-2.8);
			\node at (0.45,-2.5) {\small $H(\Sigma,\theta)$};
			\draw[thick, ->] (-1.2,-2.5) -- (-1.5,-2.5) -- (-1.5,0) -- (-1.2,0);
		\end{tikzpicture}
		\hspace{0.08\textwidth}
		\begin{tikzpicture}[scale=1.1]
			\draw[thick] (0.55,-0.3) rectangle (-0.55,0.3);
			\node at (0,0) {\small $d$};
			\draw[thick, red, ->] (0,-0.3) -- (0,-0.9);
			\draw[thick] (0.55,-0.9) rectangle (-0.55,-1.5);
			\node at (0,-1.2) {\small $\hat{s}(Z)$};
			\draw[thick] (-0.65, 0.4) rectangle (1.4, -1.7);
			\node at (1.0,-0.6) {\small ML};
			\draw[thick, ->] (0,-1.5) -- (0,-2.1);
			\draw[thick] (3.4,-1.2) ellipse (1.6 and 0.4);
			\node at (3.4,-1.2) {\small $(\theta(Z), \Sigma(Z))$};
			\draw[thick, ->] (3.4,-1.6) -- (3.4,-2.1);
			\draw[thick] (-0.6,-2.1) rectangle (4.4,-2.7);
			\node at (1.7,-2.4) {\small $H(\theta(Z), \Sigma(Z))$};
			\draw[thick, ->] (4.4,-2.4) -- (5.3,-2.4) -- (5.3,-1.2) -- (5,-1.2);
		\end{tikzpicture}
		\caption[]{Graphical comparison of BLP (left panel) and NAME (right panel). Computationally heavy part of forming the predicted shares is shown by the red arrow. In BLP this step is redone on every loop of the optimization by inferring the shares from the structural equations. In NAME, in contrast, this step is only done once outside of the optimization over structural parameters.}
		\label{fig:schematic_comparison}
	\end{center}
\end{figure}

Next, we need to explain how to extend the results from a single $Z_0$ to an arbitrary $Z$. To do so, we assume that a researcher ran a procedure described above multiple times for a number of  values of $Z$'s. As a result, she ends up with a set of $\{(Z_1,\theta_1,\xi_1),\dots,(Z_K,\theta_K,\xi_K)\}$. This set can be used to run a learning algorithm of $\theta$'s and $\xi$'s on $Z$'s, leading to full functions $\theta(Z)$ and $\xi(Z)$. For special cases, one can derive alternative and more computationally efficient ways to extend from $Z_0$ to arbitrary $Z$. We now explore some of the special cases.

\subsubsection{Special case 1: \cite{berry1994estimating}}
So far we have estimated $\hat{\beta}(Z)$ for a single value of $Z=Z_0$. However, $\hat{\beta}(Z)$ can be easily derived from $\hat{s}(Z,j,m)$ by solving a system of linear equations. As before, let $\beta = \beta(Z_0)$. Additionally, define $c_{jm}(Z)=\left(\theta(Z) - \theta\right)^TX_{jm}$. Then, $\theta(Z)^TX_{jm}=\theta^TX_{jm} + c_{jm}(Z)$ and
\begin{eqnarray*}
  s_{jm}(Z_{im}) &=& \dfrac{w_{jm}(Z_{im})\exp{\theta^TX_{jm}+\xi_{jm}}}{1 + \sum_{j'}w_{j'm}(Z_{im})\exp{\theta^TX_{j'm}+\xi_{j'm}}},
\end{eqnarray*}
where $w_{jm}(Z) = \exp{c_{jm}(Z)}$. 

Note that the set of equations
\begin{eqnarray*}
\hat{s}_{jm}(Z_0) &=& \frac{\exp{\hat{\theta}^TX_{jm}+\hat{\xi}_{jm}}}{1 + \sum_{j'}\exp{\hat{\theta}^TX_{j'm}+\hat{\xi}_{j'm}}}
\end{eqnarray*}
for $j=1,\dots,J$ and $m=1,\dots,M$ uniquely defines the values of $\exp{\hat{\theta}^TX_j+\hat{\xi}_{jm}}$ for all $j$ and $m$. Let their estimates produced from $\hat{s}_{jm}(Z_0)$ be $\hat{E}_{jm}$. Then, we have the following system of linear equations:
\begin{eqnarray*}
  \hat{s}_{jm}(Z) + \hat{s}_{jm}(Z)\sum_{j'=1}^J\hat{w}_{j'm}(Z)\hat{E}_{j'm} &=& \hat{w}_{jm}(Z)\hat{E}_{jm}
\end{eqnarray*}

\noindent We can solve for $\hat{w}_{jm}(Z)$, which is equivalent to $\hat{c}_{jm}(Z)$. Finally, once we have $\hat{c}_{jm}(Z)$ for all $j$ and $m$, we can use these values to recover $\hat{\beta}(Z)$ by regressing $\hat{c}_{jm}(Z)$ on $X_{jm}$.

\subsubsection{Special case 2: \cite{doi:10.1086/379939}}
The overall idea is similar to the case of \cite{berry1994estimating} with the only difference being that the resulting system of equations for $\hat{w}_{jm}(Z)$ is not linear. To overcome this problem, one can use BLP contraction mapping as an efficient way of solving a resulting system. Indeed,
\begin{eqnarray*}
	s_{jm}(Z) &=& \int\dfrac{\exp{c_{jm}(Z)+\theta^TX_{jm}+\nu_iX_{jmt}+\xi_{jm}}}{1 + \sum_{j'}\exp{c_{j'm}(Z)+\theta^TX_{j'm}+\nu_iX_{j'm}+\xi_{j'm}}}\,dF(\nu_i).
\end{eqnarray*}
Plugging in the estimated values, we obtain the following system on $c_{jm}(Z_i)$:
\begin{eqnarray*}
	\hat{s}_{jm}(Z) &=& \int\dfrac{\exp{c_{jm}(Z)+\hat{\theta}(Z_0)^TX_{jm}+\nu_iX_{jmt}+\hat{\xi}_{jm}}}{1 + \sum_{j'}\exp{c_{j'm}(Z)+\hat{\theta}(Z_0)^TX_{j'm}+\nu_iX_{j'm}+\hat{\xi}_{j'm}}}\,d\hat{F}(\nu_i).
\end{eqnarray*}
\cite{berry1995automobile} show that the mapping:
\begin{eqnarray*}
	x\mapsto \left(\int\dfrac{\exp{x_j+r_j+v_{ij}}}{1+\sum_{j'}\exp{x_{j'}+r_{j'}+v_{ij'}}}\,dF(v_i)\right)_{j=1}^J
\end{eqnarray*}
is a contraction mapping. As a result, one can solve the system for $c_{jm}(Z)$ efficiently by applying the BLP-contraction mapping iteratively.

\section{Consistency}

The first theoretical result we present is the consistency of the estimator $\hat{\theta}(Z)$. Informally, the result states that if we have any consistent estimator of $s_{jm}(Z_0)$, the resulting procedure leads to a consistent estimator of $\theta(Z_0)$. Consequently, many of the known flexible machine learning algorithms for non-parametric probability estimation result in consistent estimators of $\theta(Z)$.

\begin{theorem}\label{thm:1}
  If $\hat{s}_{m}(Z_0)\Convprob s_{m}(Z_0)$ uniformly over $m = 1,\dots,M$ as $M\longrightarrow\infty$, then $\hat{\theta}_M(Z_0)\Convprob\theta(Z_0)$.
\end{theorem}

\noindent The proof is given in Appendix \ref{app:proof:thm:1}. There are a number of estimation procedures that guarantee the consistency required by Theorem \ref{thm:1}. The following result provides several sufficient conditions.

\begin{lemma}
The following procedures provide a (uniformly over $m$) consistent estimator of $s_{jm}(Z_0)$:
\begin{enumerate}[label=(\roman*)]
  \item Kernel ridge regression with the tuning parameter approaching zero.
  \item Sieve estimator with the limiting space containing $s(Z)$.
  \item AdaBoost with the appropriate stopping rule.
  \item Sufficiently flexible neural network.
\end{enumerate}
\end{lemma}
\begin{proof}
  The result for (i) follows from \cite*{evgeniou2000regularization} and Theorem 29.8 of \cite*{devroye2013probabilistic}. \cite*{chen2007large} provides the technical conditions under which (ii) produces a consistent estimator. The results from \cite*{bartlett2007adaboost} specify the stopping rule that depends on the sample size and guarantees the convergence. \cite*{cybenko1989approximation} provides the results for (iv) for the case of a single hidden layer neural network with a sufficiently large number of neurons.
\end{proof}

\section{Distributional properties}
\subsection{Asymptotic distributions}
\begin{theorem}
	\label{thm:2}
	Let $\psi_m[Z_0] = H_{\xi,m}\xi_{s,m}(\hat{s}_m-s_m)[Z_0]$, $\psi[Z] = \bar{\psi}_m[Z]$, where $H_{\xi,m} = \dfrac{\partial}{\partial\xi}H_m$ and $\xi_{s,m}(\hat{s}_m-s_m) = \dfrac{\partial}{\partial s}\xi_m$, and square brackets stand for evaluation point. If:
	\begin{enumerate}[label=(\roman*)]
		\item $\dfrac{1}{\sqrt{M}}\sum\limits_{m=1}^M\E\psi_m\Convprob0$,
		\item $\dfrac{1}{M}\sum\limits_{m=1}^M\V\psi_m\Convprob0$, and
		\item $(\bar{H}_{\theta}+\overline{H_{\xi}\xi_{\theta}})$ is invertible,
	\end{enumerate}
	then $\sqrt{M}(\hat{\theta}-\theta)[Z] \Convprob -\sqrt{M}(\bar{H}_{\theta}+\overline{H_{\xi}\xi_{\theta}})^{-1}\bar{H}[Z]$.
\end{theorem}

\noindent The proof is given in Appendix \ref{app:proof:thm:2}. There are a few important corollaries that follow from this result.
\begin{corollary}
	Fix any set $\mathcal{Z}=\{Z_1,\dots,Z_R\}$ such that for any $Z_r\in\mathcal{Z}$ conditions of the theorem above hold. Then $\sqrt{M}(\hat{\theta}-\theta)[\mathcal{Z}] \Convprob -\sqrt{M}(\bar{H}_{\theta}+\overline{H_{\xi}\xi_{\theta}})^{-1}\bar{H}[\mathcal{Z}]$.
\end{corollary}
\begin{proof}
	Immediately follows from the fact that convergence in probability for individual variables implies convergence for the vectors.
\end{proof}

\begin{corollary}
	Sufficient conditions for $\psi[Z]$ to satisfy (i) and (ii) are respectively:
	\begin{enumerate}[label=(\roman**)]
		\item $H_{\xi}\xi_{s}$ is bounded and there exists $h_1(M)=o_p(\frac{1}{\sqrt{M}})$ such that:
		\begin{eqnarray*}
			\plim\,\dfrac{\Big|\left\{m\,:\,|\E\hat{s}_m-s_m|>h_1(M)\right\}\Big|}{\sqrt{M}} = 0
		\end{eqnarray*}
		\item $H_{\xi}\xi_{s}$ is bounded, $\hat{s}_m$ are independent, and there exists $h_2(M)=o_p(1)$ such that:
		\begin{eqnarray*}
			\plim\,\dfrac{\Big|\left\{m\,:\,\V\hat{s}_m>h_2(M)\right\}\Big|}{M} = 0
		\end{eqnarray*}
	\end{enumerate}
\end{corollary}
\begin{proof}
	$(i^*)$ If $H_{\xi}\xi_{s}$ is bounded, then $\dfrac{1}{\sqrt{M}}\sum\limits_{m=1}^M|\E\hat{s}_m-s_m|\Convprob0$ implies condition (i). Note that both $\hat{s}_m$ and $s_m$ are also bounded by the vector of 1's, hence:
	\begin{align*}
		&\dfrac{1}{\sqrt{M}}\sum\limits_{m=1}^M|\E\hat{s}_m-s_m|\le\\
		&\quad\dfrac{1}{\sqrt{M}}Mh_1(M)+\dfrac{1}{\sqrt{M}}\Big|\left\{m\,:\,|\E\hat{s}_m-s_m|>h_1(M)\right\}\Big|\cdot(1,\dots,1)\\
		&\quad\Convprob\sqrt{M}h_1(M)+o_p(1) = o_p(1) 
	\end{align*}
	$(ii^*)$ If $H_{\xi}\xi_{s}$ is bounded, then $\dfrac{1}{M}\sum\limits_{m=1}^M\V\hat{s}_m\Convprob0$ implies condition (ii). Then:
	\begin{align*}
		&\dfrac{1}{M}\sum\limits_{m=1}^M\V\hat{s}_m\le\\
		&\quad\dfrac{1}{M}Mh_2(M)+\dfrac{1}{M}\Big|\left\{m\,:\,\V\hat{s}_m>h_2(M)\right\}\Big|\cdot(1,\dots,1)\\
		&\quad\Convprob h_2(M) + o_p(1) = o_p(1)
	\end{align*}
\end{proof}

\noindent It is worth discussing the last corollary in a bit more detail. What it stays formally can be interpreted as ``the first step estimation is not the biggest concern.'' Condition $(i^*)$ says that biases should not be too large for most of the markets. In particular, once a researcher undersmoothes at the first stage, the bias term is 0, and condition $(i^*)$ is mechanically satisfied. Another common application would be $N_m>>M$ for every $m$. In this case, the bias term, which is driven by $N_m$, has lower order compared to $M$. Note that up to $\sqrt{M}$ of share estimates can be arbitrarily biased or even inconsistent. Condition $(ii^*)$ simply says that the rate of convergence at the first step is also irrelevant as long as it is uniformly consistent. In particular, undersmoothing can suffer from low convergence rates due to higher variance. Yet, it does not affect the precision of the estimation in the limit. One common situation in the applied research could be a situation when $s_m$ is assumed to belong to a particular smoothness class over a compact support, while for some $c>0$ and $N(M)\to\infty$ we have $cN(M)\le N_m$ for all $m$. 

\subsection{Oracle property}
This subsection analyzes the relationship between the proposed estimator and the estimator that uses known shares instead. As we show, the asymptotic distribution of the two coincide, meaning the proposed estimator has the \textit{oracle property}.

\begin{theorem}
	Let $\hat{\theta}^o(Z)$ be the estimator obtained using the actual market shares. Then $\sqrt{M}(\hat{\theta}^o-\theta)[Z] \Convprob -\sqrt{M}(\bar{H}_{\theta}+\overline{H_{\xi}\xi_{\theta}})^{-1}\bar{H}[Z]$.
\end{theorem}

\begin{proof}
	By definition of $\hat{\theta}^o[Z]$:
	\begin{eqnarray*}
		\dfrac{1}{M}\sum\limits_{m=1}^MH(\hat{\theta}^o,\hat{\xi}(s_m,\hat{\theta}^o,X_m)),X_m,W_m)[Z] = 0
	\end{eqnarray*}
	Using the mean value theorem, for some $\tilde{\theta}\Convprob\theta$:
	\begin{align*}
		&\dfrac{1}{M}\sum\limits_{m=1}^MH(\theta,\hat{\xi}(s_m,\theta,X_m),X_m,W_m)[Z] + \\
		&\quad\dfrac{1}{M}\sum\limits_{m=1}^MH_{\theta}(\tilde{\theta},\hat{\xi}(s_m,\tilde{\theta},X_m),X_m,W_m)(\hat{\theta}-\theta)[Z] +\\
		&\quad\dfrac{1}{M}\sum\limits_{m=1}^{M}H_{\xi}(\tilde{\theta},\hat{\xi}(s_m,\tilde{\theta},X_m),X_m,W_m)\hat{\xi}_{\theta}(s_m,\tilde{\theta},X_m)(\hat{\theta}-\theta)[Z] = 0.
	\end{align*}
	Solving for $\sqrt{M}(\hat{\theta}-\theta)[Z]$, we get:
	\begin{eqnarray*}
		\sqrt{M}(\hat{\theta}-\theta)[Z] \Convprob -\sqrt{M}(\bar{H}_{\theta}+\overline{H_{\xi}\xi_{\theta}})^{-1}\bar{H}[Z].
	\end{eqnarray*}
\end{proof}

\subsection{Asymptotic normality}
This subsection shows that under some extra assumptions on the DGP, we can get asymptotic distribution of $(\hat{\theta}-\theta)[\mathcal{Z}]$ in a closed form. Moreover, this distribution turns out to be a Gaussian process over $\mathcal{Z}$. 
\begin{theorem}
	Assume that DGP is such that $\E[H(Z)^2] = \sigma^2(Z)$ is finite for every $Z$ and conditions of Theorem \ref{thm:1} are satisfied. Then, the asymptotic distribution for $\sqrt{M}(\hat{\theta}-\theta)$ is a Gaussian process with mean 0 and the covariance function given by:
	\begin{eqnarray*}
		V(Z_1,Z_2) = R(Z_1)^{-1}\E[H(Z_1)H'(Z_2)](R(Z_2)^{-1})^T,
	\end{eqnarray*}
	where $R(Z) = (\E[H_{\theta}(Z)]+E[H_{\xi}\xi_{\theta}(Z)])$
\end{theorem}
\begin{proof}
	Recall that $\sqrt{M}(\hat{\theta}-\theta)[\mathcal{Z}] \Convprob -\sqrt{M}(\bar{H}_{\theta}+\overline{H_{\xi}\xi_{\theta}})^{-1}\bar{H}[\mathcal{Z}]$. The asymptotic distribution of $\sqrt{M}(\hat{\theta}-\theta)[\mathcal{Z}]$ is thus the same as that of $-\sqrt{M}(\bar{H}_{\theta}+\overline{H_{\xi}\xi_{\theta}})^{-1}\bar{H}[\mathcal{Z}]$. The only asymptotically random term in this expression is $\sqrt{M}\bar{H}[\mathcal{Z}]$. Fix any $(Z_1,\dots,Z_r)$. The CLT for vectors implies that vector $\sqrt{M}\bar{H}[(Z_1,\dots,Z_r)]$ will have asymptotically normal distribution. Hence, the asymptotic distribution of the process $\sqrt{M}\bar{H}[\mathcal{Z}]$ is indeed Gaussian. We thus only need to characterize its mean and covariance function. We start with the mean:
	\begin{eqnarray*}
		a\E\sqrt{M}\bar{H}[Z] = \E H(Z) = 0,\quad \forall Z\in\mathcal{Z}.
	\end{eqnarray*}
	For the covariance:
	\begin{eqnarray*}
		a\cov\left(\sqrt{M}\bar{H}(Z_1),\sqrt{M}\bar{H}(Z_2)\right) \Convprob \E[H(Z_1)H^T(Z_2)].
	\end{eqnarray*}
\end{proof}

\section{Generalizations and extensions}
\subsection{Aggregate moments}
With an aggregate moment condition the full estimation problem can be written as:
\begin{align*}
  &\E_{(W,X)}\left[H\left(\theta(Z);s(Z),X,W\right)\right] = 0, \quad \forall Z\\
  &\E_{(W,X,Z)}\left[G\left(\gamma(Z);s(Z),X,W\right)\right] = 0,
\end{align*}
where $H(\cdot)$ are the individual moments, and $G(\cdot)$ is the aggregate moment. We assume that $\gamma\subset\theta$, that is $\theta = (\gamma, \gamma')$ with $\gamma(Z)$ being of much smaller dimension than $\theta(Z)$. In the demand estimation setup, elasticities only depend on the parameters through $\alpha$ and $s$, i.e., $\alpha(Z)$ is one-dimensional no matter how high the dimensionality of $\theta(Z)$ is. To handle the problem, we first fix some finite base set $\mathcal{Z}=\{Z_1,\dots,Z_K\}$ and some interpolation method providing us $f(Z)$ for any $Z$ once we have values $f_1,\dots,f_K$ for $\mathcal{Z}$. We can then approximate the problem above as:
\begin{align*}
  &\E_{(W,X)}[H(\theta(Z);s(Z),X,W)] = 0, \quad \forall Z\in\mathcal{Z}\\
  &\E_{(W,X,Z)}[G(\gamma(Z);s(Z),X,W)] = 0
\end{align*}
Next, we assume that the loss function has the form:
\begin{equation*}
  \mathcal{L}(\theta) = \dfrac{1}{K}\sum\limits_{k=1}^{K}\bar{H}_k^TR_k\bar{H}_k + \underbrace{h}_{\text{relative weight}}\bar{G}^TR_0\bar{G},
\end{equation*}
where $h>0$ is the constant allowing the researcher to scale up and down the importance of the aggregate moment. The key observation that allows us to greatly simplify the optimization is the fact that the gradient is hard to compute for $\gamma(Z)$ responsible for only a small fraction of the parameters. Differentiating w.r.t.\ $\gamma'$ we obtain:
\begin{equation*}
  \dfrac{\partial\mathcal{L}}{\partial \gamma'_k} = \dfrac{2}{K}\dfrac{\partial \bar{H}_k^T}{\partial \gamma_k'}R_k\bar{H}_k
\end{equation*}
which is the same exact form as the one for the problem without the aggregate moment. Note that it depends on a single $k$ and only a small number of parameters, $\theta_k$, as $\bar{H}_k$ depends on $\theta_k$ only. The last piece is to differentiate the loss function w.r.t.\ $\gamma$. This can be done numerically, as the number of parameters in $\gamma$ is small. After obtaining $\partial\mathcal{L}/\partial \theta$, one can use a gradient descent method (with any machine learning stabilization technique, such as ADAM, RMSprop, etc.) to find optimal values of parameters.

\subsection{Shape assumptions}
The procedure allows the researcher to add extra assumptions on the shape of the functions $\theta(Z)$ and $\xi(Z)$. Those assumptions can be helpful if the dimensionality of $Z$ is large. Here we introduce the sparse design and study its properties. The design assumes that $\theta$ and \textit{all} of $\xi_m$'s only depend on at most $p_0$ of (the same) characteristics. The sparse design is suitable for demand estimation in many IO applications and may be a modeling choice in other disciplines.
\begin{definition}
	The model is $p_0$-sparse if there exists a partition of $Z=(U,U')$ such that:
	\begin{enumerate}[label=(\roman*)]
		\item $U$ is $p_0$-dimensional,
		\item $\forall\,u, u', u''$ such that $(u,u'),(u,u'')\in\mathcal{Z}$ we have $\theta(u,u')=\theta(u,u'')$ and $\xi_m(u,u')=\xi_m(u,u'')$ for every $m$.
	\end{enumerate} 
\end{definition}
\noindent The sparsity condition guarantees that even if the dimensionality of $Z$ is huge, the precision of the model can still be high. Indeed, the researcher only needs to identify a few variables among $Z$ that actually affect the preferences. Before stating the main result for the sparse design, we need to state a lemma that links the shape conditions on $\theta$ and $\xi$ to the shape conditions on $s$.
\begin{lemma}
	If the design is $p_0$-sparse, then $s_m$ is also $p_0$-sparse with respect to the same partition $(U,U')$.
\end{lemma}
\begin{proof}
	Recall that $s_m(Z) = s_m(\theta(Z),\xi_m(Z),X_m)$. Now $\forall\, u, u', u''$ such that $(u,u'),(u,u'')\in\mathcal{Z}$:
	\begin{eqnarray*}
		s_m(u,u') = s_m(\theta(u,u'),\xi_m(u,u'),X_m) = s_m(\theta(u,u''),\xi_m(u,u''),X_m) = s_m(u,u'').
	\end{eqnarray*}
\end{proof}
\begin{theorem}
	\label{thm:3}
	Without loss of generality let $\E[Z]=0$. Assume that $(p,\{N_m\}_1^M,U,U',c_1)$ and the corresponding $p_0$-sparse designs satisfy the following conditions:
	\begin{enumerate}[label=(\roman*)]
		\item Define $U_m = \{u\in U:|\E[d_{im}Z_{im}^u]|\ge c_1\}$. Then $U=\bigcup\limits_{m}U_m$,
		\item $\log(p) = o(\min(N_m))$,
		\item $M = o(\min(N_m))$,
		\item $\dfrac{c_1^2N_m}{8}-\log(p_0)-\log(M)\to\infty$.
	\end{enumerate}
	Then, the following estimator converges in probability to the estimator $\hat{\theta}(U)$, i.e., if $U$ was known:
	\begin{enumerate}[label=(\roman*)]
		\item Recover $\hat{U}$ by the following algorithm:
		\begin{enumerate}
			\item For every market $m$ and every variable $Z^u$ add $u$ to $\hat{U}_m$ iff: 
			\begin{eqnarray*}
				\dfrac{1}{\sqrt{\hat{V}_{mu}N_m}}\Big|\sum\limits_{i=1}^{N_m}s_{im}Z_{im}^u\Big|\ge \sqrt{2(\log p + \log N_m)},
			\end{eqnarray*}
			\item Set $\hat{U}=\bigcup\limits_{m}\hat{U}_m$.
		\end{enumerate}
		\item Estimate $\hat{\theta}(Z)$ as $\hat{\theta}(\hat{U})$ using NAME.
	\end{enumerate}
\end{theorem}
\noindent The proof is provided in Appendix \ref{app:proof:thm:3}.
\begin{remark}
	The procedure described above is an application of a fairly general idea of sparse model estimation to this setup and NAME. This idea boils down to a two-step estimation. At the first step, we recover the set of confounding variables. At the second step, the main estimation is done using only the recovered subset of the variables. Applications to different settings can be found in \cite{chernozhukov2017double} and \cite{gillen2015blp}.
\end{remark}

\section{Simulations}
\subsection{Functional form mis-specification}\label{sec:sims}
To illustrate the performance differences of different methods we consider a setting with $J=2$ goods (plus the outside good) and $M=50$ separate markets. Each market is populated by $N=1000$ consumers. Both $X_{jm}$ and $Z_{im}$ are one-dimensional ($p=k=1$) and the price is exogenous.\footnote{This is done for simplicity to ignore the instruments, $W_{jm}$, but does not affect the results.}

We let $\beta(Z) = \gamma_0 + \gamma_1Z + \gamma_2Z^2$ and consider three different cases: (i) the standard approach when $\beta(Z)$ is misspecified and $\tilde{\beta}(Z) = \tilde{\gamma}_0 + \tilde{\gamma}_2Z^2$ is used instead, (ii) the standard approach when $\beta(Z)$ is correctly specified, and (iii) the proposed algorithm. The true value of the price coefficient, $\alpha$, is $1$.


We use the minimum distance estimator with the following moments:
\begin{enumerate}
  \item The covariance between $X_{jm}$ and $\hat{\xi}_{jm}$ has to be close to zero.
  \item The covariance between $P_{jm}$ and $\hat{\xi}_{jm}$ has to be close to zero.
  \item The covariance between $Z_{im}$ and $X_{d_{im}m}$ (where $X_{d_{im}m}$ denotes the characteristics of the good chosen by individual $i$ in market $m$) has to be close to the covariance between $Z_{im}$ and $\sum_j \hat{s}_{jm}(Z_{im})X_{jm}$ (only for the oracle case).
  \item The covariance between $Z_{im}^2$ and $X_{d_{im}m}$ has to be close to the covariance between $Z_{im}^2$ and $\sum_j \hat{s}_{jm}(Z_{im})X_{jm}$ (only for the misspecified and the oracle cases).
\end{enumerate}

\noindent For the first step of the algorithm we use the kernel ridge regression as presented in \cite*{Murphy2012MachineL}.

We estimate the model $B=50$ times and obtain an estimate $\hat{\alpha}_b$ at each iteration $b=1,\dots,B$. The distributions of these values are shown in Figure \ref{fig:sims_alphas}. In the plot, ``Misspecified'' refers to case (i), ``Oracle'' to case (ii), and ``Proposed'' to case (iii).

\subsubsection{Precision}
As reported in Table \ref{tbl:rmse}, the proposed procedure performs well compared to the oracle case in terms of both---the bias and the root-mean-square error---while the misspecified procedure produces the estimates that are much more biased.
\begin{table}[!h]
\caption{Bias and Root-Mean-Square Error}
\label{tbl:rmse}
\begin{center}
  \begin{tabular}{l|c|c|c}
    \hline\hline
     \multicolumn{4}{c}{Bias and Root-Mean-Square Error of the Algorithms} \\ \hline
     & Misspecified & Proposed & Oracle \\ \hline
     Bias & $0.42$ & $0.01$ & $0.07$ \\
     RMSE & $0.42$ & $0.07$ & $0.21$ \\ \hline\hline
  \end{tabular}
\end{center}
\end{table}

\begin{figure}[!h]
\begin{center}
 \includegraphics[width=\textwidth]{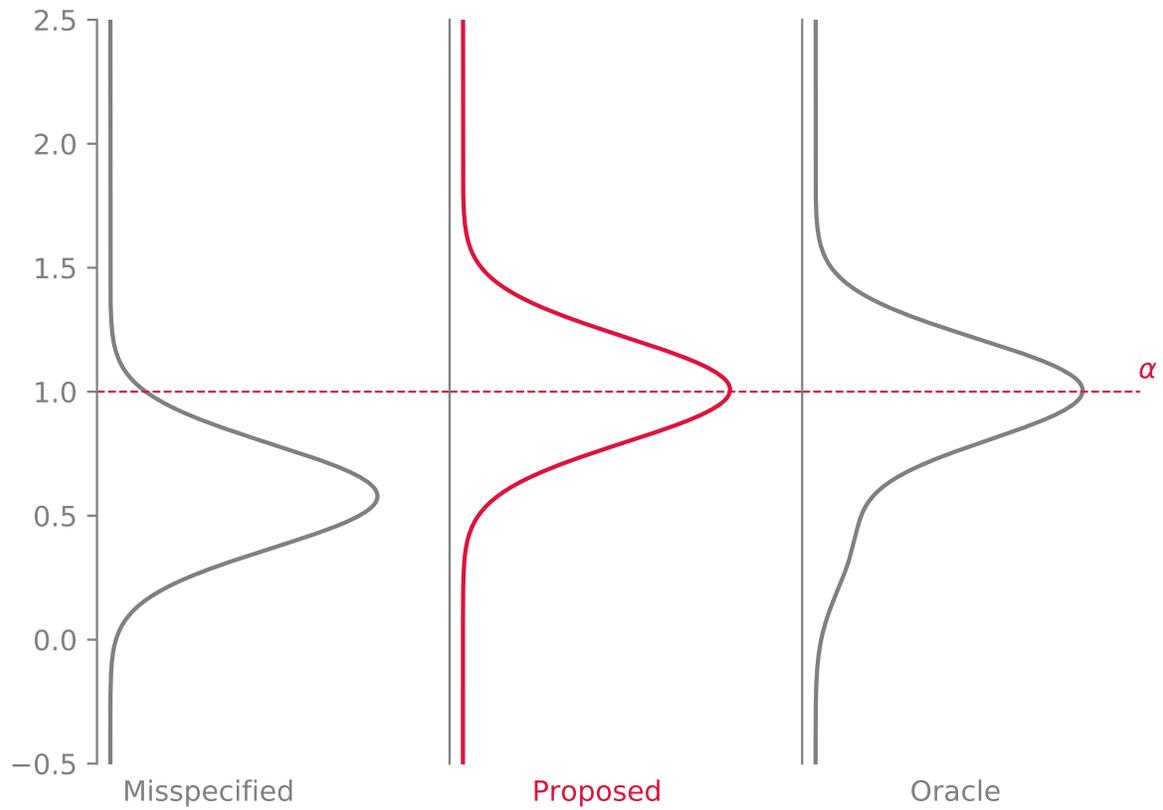}
\end{center}
\caption[]{Simulated Distributions of the Price Coefficient as Described in Section \ref{sec:sims}}
\label{fig:sims_alphas}
\end{figure}


\subsubsection{Computational Costs}
One of the main advantages of our procedure is its computational simplicity. Even for a single dimensional individual level characteristic, $Z$, it converges on average almost two times faster than the misspecified procedure (which requires just a single additional parameter) and almost four times faster than the correct specification (two additional parameters). When the dimensionality of $Z$ or the complexity of $\beta(Z)$ as a function of $Z$ increase, the difference becomes even more substantial.

\begin{table}[!h]
\caption{Computational Costs}
\label{tbl:costs}
\begin{center}
  \begin{tabular}{l|c|c|c}
    \hline\hline
     \multicolumn{4}{c}{Computational Times of the Algorithms (in seconds)} \\ \hline
     & Misspecified & Proposed & Oracle \\ \hline
     Average time & $191$\,s & $107$\,s & $406$\,s \\ 
     Standard deviation & $164$\,s & $105$\,s & $317$\,s \\ \hline\hline
  \end{tabular}
\end{center}
\end{table}

\subsection{Sparse design}
In our basic setup, we simulate the simplest model to demonstrate how subset recovery works. There are $M=50$ homogenous markets each populated with $N_m=1000$ individuals. The individuals choose whether to purchase a single available good, or stick to the outside option, so $J=1$. The good has only one characteristic, meaning that $K=1$. The space of observed characteristics is 1000-dimensional, but in every market only one of the two first components affects the preferences with equal probability. This means that $p=1000$, $p_0=2$, and $U_m = \{Z_1\}$ or $U_m=\{Z_2\}$ equally likely. Denote by $Z_{r_m}$ the variable influencing choices in the $m$th market. The utility of the consumer is then $u_{i1m} = Z_{r_m}X_m + \epsilon_{i1m}$ while $u_{i0m}=\epsilon_{i0m}$. The choice is made according to the discrete choice model. All $X$'s and $Z$'s are $i.i.d.$ standard normal, while $\epsilon$'s are $i.i.d.$ EV-I random variables. 

We first run 1000 simulations for this particular set of parameters. The results are summarized in the first row of Table \ref{tab:sparse:summary:1}. Note that the vast majority of simulations resulted in a correct recovery of $U$ with only 1\% being too conservative. No simulations resulted in under-recovery of $U$. The second row of Table \ref{tab:sparse:summary:1} runs the same simulations but for a much larger $p_0=20$ with 5 variables chosen at random kicking in for each individual market. The correct recovery still occurs in the vast majority of the cases accounting for 96.7\% of all simulations with 0.6\% of simulations resulting in an overrecovery of size one. The main difference with $p_0=2$ case occurs due to 2.7\% of underrecoveries. Those underrecoveries, however, are seldom of size 2 and are never larger. All simulations resulted in one of the three cases: exact recovery, underrecovery, or overrecovery. 

\begin{table}
	\caption{Simulation Results for Sparse Design}
	\label{tab:sparse:summary:1}
	\begin{center}
	\begin{tabular}{c|c|ccc|ccc|c|c}
		\hline\hline
		& \multicolumn{1}{|c|}{$\hat{U}=U$} &\multicolumn{3}{|c|}{$U\subsetneq \hat{U}$} & \multicolumn{3}{|c|}{$\hat{U}\subsetneq U$} & Other & Total \\
		$|U\bigtriangleup\hat{U}|$&  & 1 & 2 & 3+ & 1 & 2 & 3+ & \\
		\hline
		$p_0=2$ & 990 & 10 & 0 & 0 & 0 & 0 & 0 & 0 & 1,000\\
		$p_0=20$ & 967 & 6 & 0 & 0 & 26 & 1 & 0 & 0 & 1,000 \\
		\hline\hline
	\end{tabular}
	\end{center}
	\begin{justify}
		\emph{Note:} $\hat U$ is the estimated set of influential characteristics, and $U$ is the true set. The total number of characteristics is $p=1,000$. $p_0$ is the number of characteristics that have non-zero effect on the preferences. $\Delta$ denotes for a symmetric difference, i.e., the set of characteristics that belong to exactly one set, and $|\cdot|$ denotes the cardinality of a set.
	\end{justify}
\end{table}	

\section{Conclusion}
The method of estimation that we propose is likely to outperform the approach commonly used in the literature unless the correct functional form is known to the researcher. We also show that the proposed estimator has nice large sample properties and is more computationally efficient than the commonly used alternatives.

There are several potential directions for future research. First, the optimal choice of $Z_0$ should be investigated. As different markets may be heterogeneous in terms of the distributions of $Z$, the precision of the estimates may suffer from the choice of $Z_0$ that is the same for every market. Second, most of the machine learning procedures require the choice of a tuning parameter (or parameters) that was largely ignored in this paper. Third, the properties of the relevant statistical tests may be investigated. Finally, it would be useful to see how the proposed method performs in applications in comparison with the alternative procedures.

\bibliographystyle{chicago}
\bibliography{dd}

\appendix
\section{Proof of Theorem 1}
\label{app:proof:thm:1}
\begin{proof}
Note that BLP inversion of the form $\hat{\xi}(s;\theta)$ for fixed parameters $\theta$ is a continuous function of $s$. Hence, if $\hat{s}_{m}(Z_0)\Convprob s_{m}(Z_0)$, we have $\hat{\xi}(\hat{s}_{m}(Z_0);\theta(Z_0))\Convprob \hat{\xi}(s_{m}(Z_0);\theta(Z_0))$ uniformly. Denote by $\Xi_{m}(\theta(Z_0))\Convprob 0$ the difference between $\hat{\xi}(\hat{s}_{m}(Z_0);\theta(Z_0))$ and $\hat{\xi}(s_{m}(Z_0);\theta(Z_0))$. The moment driven objective function satisfies the condition:
  \begin{eqnarray*}
    \dfrac{1}{M}\sum\limits_{m=1}^MH(\theta(Z_0); \xi_{m}(Z_0), X_{m}, W_{m})\Convprob 0
  \end{eqnarray*}
  and has two important properties: (i) it is continuous in $\xi$, (ii) it satisfies the identification assumption,\footnote{For the proof that steps 4--5 of the proposed procedure lead to consistent estimators as $M\longrightarrow\infty$ see, for example, \cite*{freyberger2015asymptotic}.} meaning that:
  \begin{eqnarray*}
    \dfrac{1}{M}\sum\limits_{m=1}^MH(\theta; \hat{\xi}(s_m;\theta), X_{m}, W_{m})\Convprob 0, \quad \text{if } \theta = \theta_0,
  \end{eqnarray*}
  and converges to non-zero value otherwise. We thus have:
  \begin{multline*}
  	\dfrac{1}{M}\sum\limits_{m=1}^MH(\theta; \hat{\xi}(\hat{s}_m;\theta), X_{m}, W_{m}) = \\
  	\dfrac{1}{M}\sum\limits_{m=1}^MH(\theta; \hat{\xi}(s_m;\theta) + \Xi_m(\theta(Z_0)), X_{m}, W_{m})\Convprob \\
    \lim\limits_{M\to\infty}\dfrac{1}{M}\sum\limits_{m=1}^MH(\theta; \hat{\xi}(s_m;\theta), X_{m}, W_{m}) = 0
  \end{multline*}
  as $\Xi_{jm}(\theta)$ converges in probability to zero uniformly over $m$. To complete the proof we need to fix any neighborhood of $\theta(Z_0)$ and to observe that the last equation implies that $\hat{\theta}(Z_0)$ must lie inside this neighborhood for large enough $M$. 
\end{proof}

\section{Proof of Theorem 2}
\label{app:proof:thm:2}
\begin{proof}
	By definition of $\hat{\theta}[Z]$:
	\begin{eqnarray*}
		\dfrac{1}{M}\sum\limits_{m=1}^MH(\hat{\theta},\hat{\xi}(\hat{s}_m,\hat{\theta},X_m)),X_m,W_m)[Z] = 0
	\end{eqnarray*}
	Using mean value theorem, for some $(\tilde{\theta}, \tilde{s})\Convprob(\theta, s)$:
	\begin{align*}
		\dfrac{1}{M}\sum\limits_{m=1}^M&H(\theta,\hat{\xi}(s_m,\theta,X_m),X_m,W_m)[Z] + \\
		&\dfrac{1}{M}\sum\limits_{m=1}^MH_{\theta}(\tilde{\theta},\hat{\xi}(\tilde{s}_m,\tilde{\theta},X_m),X_m,W_m)(\hat{\theta}-\theta)[Z] +\\
		&\dfrac{1}{M}\sum\limits_{m=1}^{M}H_{\xi}(\tilde{\theta},\hat{\xi}(\tilde{s}_m,\tilde{\theta},X_m),X_m,W_m)\hat{\xi}_{s}(\tilde{s}_m,\tilde{\theta},X_m)(\hat{s}_m-s_m)[Z] + \\
		&\dfrac{1}{M}\sum\limits_{m=1}^{M}H_{\xi}(\tilde{\theta},\hat{\xi}(\tilde{s}_m,\tilde{\theta},X_m),X_m,W_m)\hat{\xi}_{\theta}(\tilde{s}_m,\tilde{\theta},X_m)(\hat{\theta}-\theta)[Z] = 0.
	\end{align*}
	Solving for $\sqrt{M}(\hat{\theta}-\theta)[Z]$, we get:
	\begin{eqnarray*}
		\sqrt{M}(\hat{\theta}-\theta)[Z] \Convprob -\sqrt{M}(\bar{H}_{\theta}+\overline{H_{\xi}\xi_{\theta}})^{-1}\bar{H}[Z] - (\bar{H}_{\theta}+\overline{H_{\xi}\xi_{\theta}})^{-1}\sqrt{M}\psi[Z].
	\end{eqnarray*}
	Finally, $\E[\sqrt{M}\psi] = \sqrt{M}\E\psi\Convprob0$ by (i) and $\V[\sqrt{M}\psi] = M\V\psi\Convprob0$ by (ii). This means that $(\bar{H}_{\theta}+\overline{H_{\xi}\xi_{\theta}})^{-1}\sqrt{M}\psi[Z]\Convprob0$, giving the result.
\end{proof}

\section{Proof of Theorem 5}
\begin{proof}
	\label{app:proof:thm:3}
	Note that $\dfrac{1}{\sqrt{N_m}}\sum\limits_{i=1}^{N_m}d_{mi}Z_i^u\to_d\mathcal{N}(0, V_{mu})$, or, alternatively, $\zeta_{mu} = \dfrac{1}{\sqrt{V_{mk}N_m}}\sum\limits_{i=1}^{N_m}d_{mi}Z_i^u\to_d\mathcal{N}(0,1)$ if $u\in U'$. Maximal inequalities for sub-Gaussian random variables\footnote{See \cite{highdimensionalstats} for an excellent exposition of the topic.} imply, in particular:
	\begin{eqnarray*}
		\P(\max\limits_{u\in U'}|\zeta_{mu}|>x) \le 2|U'|\exp{-\dfrac{x^2}{2}} \le 2p\exp{-\dfrac{x^2}{2}}.
	\end{eqnarray*}
	Let this probability decay as $\dfrac{2}{N_m}$. Then the corresponding $c_0$ can be found as a solution to the equation:
	\begin{eqnarray*}
		\dfrac{2}{N_m} = 2p\exp{-\dfrac{c_0^2}{2}},
	\end{eqnarray*}
	which after simple algebra can be simplified to:
	\begin{eqnarray*}
		c_0 = \sqrt{2(\log p + \log N_m)}
	\end{eqnarray*}
	We thus can bound the probability that no $u\in U'$ will be selected. We now turn to a question what is the probability of not including at least a single $u\in U_m$ for the same threshold.
	\begin{eqnarray*}
		&\P(\min\limits_{u\in U_m}|\zeta_{mu}|\le c_0) \le \P(\min\limits_{u\in U_m}\{\zeta_{mu}-\E[\zeta_{mu}]\}\le c_0-c_1\sqrt{N_m}) =\\ 
		&\P(\max\limits_{u\in U_m}\{\zeta_{mu}-\E[\zeta_{mu}]\}\ge c_1\sqrt{N_m}-c_0) \le p_0\exp{-\dfrac{(c_1\sqrt{N_m}-c_0)^2}{2}}.
	\end{eqnarray*}
	Plugging in the expressions for $c_0$ and $c_1$, we have (since $\log p = o(\min(N_m))$):
	\begin{eqnarray*}
		p_0\exp{-\dfrac{(c_1\sqrt{N_m}-\sqrt{2(\log K + \log N_m)})^2}{2}}\le p_0\exp{-\dfrac{c_1^2N_m}{8}}.
	\end{eqnarray*}
	To complete the proof, note that the probability of recovering correct $U_m$ \emph{in all markets at the same time} is at least:
	\begin{eqnarray*}
		\P(\hat{U}_1=U_1,\dots,\hat{U}_M=U_M) \ge 1 - \dfrac{2M}{\min N_m} - p_0M\exp{-\dfrac{c_1^2N_m}{8}}.
	\end{eqnarray*}
	As $M=o(\min(N_m))$, $\dfrac{2M}{\min N_m}\to0$. Also:
	\begin{eqnarray*}
		p_0M\exp{-\dfrac{c_1^2N_m}{8}} = \exp{\log(p_0)+\log(M)-\dfrac{c_1^2N_m}{8}}\to0
	\end{eqnarray*}

	Once $U_m$'s have been recovered, we can construct $\hat{U} = \bigcup\limits_{m}\hat{U}_m$ that is equal to $\bigcup\limits_{m}U_m=U$ with probability converging to 1. As a consequence $\hat{\theta}(\hat{U})\to_p\hat{\theta}(U)$.
\end{proof}

\end{document}